\documentclass[final]{ws-rmp}
\usepackage[T1]{fontenc}
\usepackage{amsmath,amscd}
\usepackage{amssymb}
\usepackage{latexsym}
\usepackage{mathrsfs}
\usepackage{makeidx}
\usepackage{subfigure}
\usepackage{lastpage}
\usepackage{fancybox}
\usepackage{lscape}
\usepackage{longtable}
\usepackage{caption}
\usepackage{dsfont}
\usepackage{accents}
\usepackage{array}
\usepackage{booktabs}
\usepackage{eucal}
\newcommand{\reg}{\textrm{reg}}

\newcommand{\Ci}{\mathcal{C}^{\infty}}

\newcommand{\loc}{\mathrm{loc}}
\newcommand{\alte}{\mathrm{alt}}
\newcommand{\ph}{\varphi}

\newcommand{\supp}{\mbox{supp}}
\newcommand{\WF}{\mbox{\rm WF}}

\newcommand{\1}{\mathds{1}}

\newcommand{\tio}{\cdot_{\calT}}

\newcommand{\calA}{{\mathcal A}}

\newcommand{\al}{\alpha}

\newcommand{\Co}{{\mathbb C}}
\newcommand{\Na}{{\mathbb N}}
\newcommand{\R}{{\mathbb R}}

\newcommand{\M}{{\mathbb M}}
\newcommand{\calT}{{\mathcal T}}
\newcommand{\calJ}{{\mathcal J}}

\newcommand{\calD}{{\mathcal D}}
\newcommand{\calE}{{\mathcal E}} 
\newcommand{\calF}{{\mathcal F}}

\newcommand{\frakA}{{\mathfrak{A}}}

\newcommand{\euB}{{\mathscr B}}
\newcommand{\euC}{{\mathscr C}}

\newcommand{\sst}[1]{\scriptscriptstyle{#1}}  
\newcommand{\starH}{\star_{\!\sst{\Delta_1}}}    
\newcommand{\minus}{\sst{-1}}   
\newcounter{saveenum} 
\begin{document}
\markboth{Katarzyna Rejzner}
{fermionic fields in the functional approach to classical field theory}
%
%
%

\title{Fermionic fields in the functional approach\\ to classical field theory}

\author{Katarzyna Rejzner}

\address{ II. Inst. f. Theoretische Physik, Universit\"at Hamburg,\\ Luruper Chaussee 149,\\
D-22761 Hamburg, Germany\\ 
\email{katarzyna.rejzner@desy.de} }

\maketitle


\begin{abstract}
In this paper, we present a formulation of the classical theory of fermionic (anticommuting) fields, which fits into the general framework proposed by R.~Brunetti, M.~D\"utsch and K.~Fredenhagen. It was inspired by the recent developments in perturbative algebraic quantum field theory and allows for a deeper structural understanding also on the classical level. We propose a modification of this formalism that allows to treat also fermionic fields. In contrast to other formulations of classical theory of anticommuting variables, we don't introduce additional Grassman degrees of freedom. Instead the anticommutativity is introduced in a natural way on the level of functionals. Moreover our construction incorporates the functional-analytic and topological aspects, which is usually neglected in the treatments of anticommuting fields. We also give an example of an interacting model where our framework can be applied.  
\end{abstract}



\section*{Introduction}
Recent developments in perturbative algebraic quantum field theory \cite{BDF,Duetsch:2000nh,DF,DF02,DF04,BreDue} have opened a new perspective also for the classical field theory. 
Using the algebraic approach one can treat the quantum algebra as a deformation of the classical structure. This view point turned out to be very successful for example in understanding the various notions of renormalization group \cite{BDF,Kai} and for applications in cosmology \cite{TCN,Thomas,Dappiaggi:2010pu}. In the second case Dirac fields play a very important role. In \cite{TCN} a suitable modification of the functional approach was developed to describe the anticommuting (fermionic) free quantum fields. The first steps to formulation of the classical theory were made as well.

The complete treatment of classical field theory of bosons shall be presented in \cite{Pedro}. Since the generalization to fermionic fields leads to some new effects, it is treated separately. In this paper we present a formalism which bases on the one of \cite{BDF} but introduces some new features characteristic to anticommuting fields. The problem of understanding the classical theory of fermions is a long standing one. Various attempts were made to tackle it. On the mathematical side there is the variational bicomplex approach \cite{Sard,Sard2,SardGia} and the supermanifold \cite{FranPol} or graded manifold \cite{CarFig97,MonVal02}  formalism. The supermanifold approach was used in \cite{Bruz-Cian-84,Bruz-Cian-86,Bruz85,Jad}. The geometrical foundations of supermechanics on graded manifolds were formulated in \cite{IboMar93,Mon92,MonVal02}, including the notion of graded Lagrangian, tangent supermanifold, space of velocities and Hamiltonian mechanics of a graded system.

The more intuitive but formal approach to the calculus of variation is presented in \cite{Witt}. It applies as well to even as to odd fields. It seems to be closer to intuition, but doesn't provide a clear mathematical structure that can be used to understand the classical theory of fermion fields. Following the spirit of \cite{BDF,Duetsch:2000nh,DF,DF02,DF04,BreDue} we propose a way to make the formal notions used in \cite{Witt} more precise, without loosing the intuition known from practical calculations in classical and quantum field theory. In contrast to most of the standard approaches we don't use the Grassman valued functions. In this respect our treatment resembles the one of \cite{Schmitt:1995fg,Schmitt:1996hs,Schmitt:1997iy}.
 The field configurations are ordinary sections of some vector bundle and the anticommutativity is introduced at the level of functionals. This way we avoid many of the technical and conceptual problems. 

The most difficult problem in understanding the classical theory of fermionic fields is the treatment of models with interaction, such as the Gross-Neveu \cite{Gross} or the Thirring model \cite{Thirring:1958in}. We propose a formalism to deal with this kind of theories in a way which agrees with the spirit of pAQFT. The classical structure presented here can be treated as a first step to quantization. In section \ref{quant} we comment further on this point and sketch the way how a deformation quantization can be performed, along the lines of \cite{BF0,BDF}. The approach presented in this paper allows not only the treatment of fermionic fields in the functional approach, but also opens a way to a locally covariant treatment of interacting models. The treatment of the free Dirac field was already presented by \cite{Ko}.

The paper is organized as follows: (1) In the first section we present an overview of all mathematical structures that will be needed for the formulation of the functional framework. We also define the kinematical structure of the theory. Since we work in the off-shell formalism, the notion of observables shall be introduced at this level without the need to solve the equations of motion. It will turn out that this is a very crucial point, especially for 
the fermionic fields.
(2) In the second section we propose a treatment of the dynamics. We construct the Poisson
structure (Peierls bracket) and intertwining maps between different Poisson structures (so called M{\o}ller maps). For the bosonic case, the existence of these maps for nonlinear actions will be proved in \cite{Pedro}.  (3) In the next section we treat the example of the Gross-Neveu model, using the functional methods. (4) Finally we give a sketch of quantization using the deformation of the Poisson structure.
\section{Algebra of functionals}
In the functional approach to classical field theory \cite{Duetsch:2000nh,BDF,DF,DF02,DF04} the basic structure is a Poisson algebra of classical observables. It is defined in a slightly abstract way as an algebra of functionals on the configuration space $E$. One starts with an off-shell setting, i.e. field configurations are not required to be solutions of some dynamical equations. They are just vector valued functions of a given type on the  Minkowski spacetime. For example in case of a scalar field these are smooth functions $E=\Ci(\M)$. We can identify classical observables with functionals on this space. Intuitively speaking, a measurement performed at a spacetime point $x$ can be associated with an evaluation functional $\Phi_x:E\rightarrow \R$, $\Phi_x(\varphi)=\varphi(x)$. Of course the full space of functionals is far too big to introduce a sensible structure on it. We have to impose some regularity conditions. First we restrict ourselves to functionals that are smooth. This notion again requires a few comments, since the calculus on infinite dimensional spaces (and $E$ is indeed infinite dimensional) is a subtle issue. Roughly speaking one endows the space $E$ with its natural locally convex topology and defines the derivative of a functional on $E$ as a generalization of the directional derivative (see \cite{Ham,Neeb} for details). We don't want to go into the details now, since in  the context of fermionic fields we will need a slightly different notion of smoothness anyway. We provide a thorough mathematical discussion of these points in section \ref{asf}. For now let us denote by $\Ci(E,\R)$ the space of smooth functionals on $E$.
Using the definition of the Peierls bracket \cite{Pei}  one  can  introduce the dynamical Poisson structure on $E$ in the Lagrange formalism. For the moment let us focus on the example of a free scalar field. The Peierls bracket of two functionals $F,G$ from a suitably chosen domain is defined as:
\begin{equation}\label{Peiscalar}
\{F,G\}:=-\langle F^{(1)},\Delta*G^{(1)}\rangle\,,
\end{equation}
where $\Delta$ is the causal propagator for the Klein-Gordon operator. We already indicated that this expression is well defined only on a suitably chosen domain. The reason for it is a singular character of the causal propagator $\Delta$. Indeed, the above expression contains implicitly a pointwise multiplication of distributions. To make sense of such an operation one has to control the singularity structure of the objects involved. The whole space $\Ci(E,\R)$ is too big to make a Poisson algebra out of it, since some functionals are simply too singular for the Peierls bracket (\ref{Peiscalar}) to be well defined. An easy way out is to consider only functionals that are compactly supported and local. The definition of the spacetime support of a smooth functional is simply a generalization of the distributional support. 
\begin{align}\label{support}
\supp\, F\doteq\{ & x\in M|\forall \text{ neighbourhoods }U\text{ of }x\ \exists \ph,\psi\in E, \supp\,\psi\subset U 
\\ & \text{ such that }F(\ph+\psi)\not= F(\ph)\}\ .\nonumber
\end{align}
The notion of locality is also quite intuitive. According to the standard definition one calls a functional $F$ local if it can be expressed as:
\begin{equation}\label{local}
F(\ph)=\int\limits_{\M} dx f(j_x(\ph))\,,
\end{equation}
where $f$ is a function on the jet space over $\M$ and $j_x(\ph)=(x,\ph(x),\partial\ph(x),\dots)$ is the jet of $\ph$ at the point $x$. It was already recognized in \cite{DF04,BDF} in the context of perturbative algebraic quantum field theory that the property of locality can be reformulated using the notion of \textit{additivity}\index{additivity of functionals} of a functional together with a certain wavefront set condition (we come back to it in section \ref{asf}). It can be shown that the Peierls bracket is well defined on the space of smooth compactly supported local functionals, but it turns out that this space is not closed under  $\{.,.\}$. To obtain a Poisson algebra one has to admit objects that are more singular. Here the microlocal analysis comes into play. Note that the wavefront set of the causal propagator is characterized as
\begin{equation}
\WF(\Delta)=\{(x,y,k,-k)\in \dot{T}^*(\M)^2|(x-y)^2=0,\ k||(x-y),\ k^2=0\}\,.
\end{equation}
By applying the H\"ormander's criterium \cite{Hoer} for multiplication of distributions one can identify a class of smooth compactly supported functionals for which (\ref{Peiscalar}) makes sense.

The construction outlined here is discussed in details in references \cite{Duetsch:2000nh,BDF,DF,DF02,DF04,BF0,BFK95}.
Here we only indicated most important features, which we will now reproduce for the case of fermionic fields. Note that the algebraic formulation allows us to work on a very abstract level and we can avoid some of the conceptual difficulties of other approaches. In the first step we have to specify our algebra of functionals. Here we encounter a first difference with respect to the bosonic case, since our functionals have to be antisymmetric in some sense. To obtain a suitable framework we have to use tools from functional analysis. Although we will need some abstract mathematics at the beginning, it turns out that once the framework is established, it can be easily applied to physical examples (section \ref{Dyn}). 

\subsection{Antisymmetric functionals}\label{asf}
%
We already indicated in the introduction to this section, that first we need to define the basic kinematical objects of the theory. Working off-shell means that at this point we don't specify the dynamics. Let $E\doteq \mathcal{E}(\M,V)$ be an infinite dimensional vector space of field configurations. Here $V$ is a $k$-dimensional complex vector space in which fields take values and $\M$ is the Minkowski spacetime with the signature $(+,-,-,-)$. Now we want to implement the notion of \textit{antisymmetry} into our kinematical structure. To this end we construct $\bigwedge\limits^\bullet E$, the exterior algebra of $E$, by taking the quotient of the tensor algebra over $E$ by an ideal $\left<x\otimes y+y\otimes x\right>$. The resulting algebra is equipped with the antisymmetric $\wedge$-product and can be written as a direct sum: $\bigwedge\limits^\bullet E=\bigoplus\limits_{p=0}^\infty\bigwedge^pE$, where $\bigwedge\limits^pE$ is a subspace of $\bigwedge\limits^\bullet E$ spanned by elements of the form: $u_1\wedge\ldots\wedge u_p$, $u_1,\ldots,u_p\in E$ and $\bigwedge\limits^0E=\Co$. 
A general element of $\bigwedge\limits^\bullet E$ is a finite sum of elements of spaces $\bigwedge\limits^pE$  (called homogenous).
Each space  $\bigwedge\limits^pE$ can be embedded in the space of antisymmetric sections \cite{TCN,Thomas}. The condition of ``asymmetry'' means that:
$u_{a_1,...,a_k,a_{k+1},...,a_p}(x_1,..., x_k,x_{k+1},..., x_p)=-u_{a_1,...,a_{k+1},a_k,...,a_p}(x_1,..., x_{k+1},x_k,..., x_p)$. We shall denote the space of antisymmetric sections by $\calE(\M^p,V^{\otimes p})^a$. This space can be equipped with the Fr\'echet topology of uniform convergence on compact subsets of $\M^p$. This induces also a topology on $\bigwedge\limits^pE$ and it follows that $\overline{\bigwedge\nolimits^pE}=\calE(\M^p,V^{\otimes p})^a$, with respect to this topology. We define now:
\begin{equation}
\mathscr{C}(\M,V)\doteq\bigoplus\limits_{p=0}^\infty\calE(\M^p,V^{\otimes p})^a\,.
\end{equation}
The dual of $\mathscr{C}(\M,V)$ is the direct product:
\begin{equation}
\calA\doteq\left(\mathscr{C}(\M,V)\right)'=\prod\limits_{p=0}^\infty\calE'(\M^p,V^{\otimes p})^a\doteq\prod\limits_{p=0}^\infty \calA^p\,.
\end{equation}
The elements of $\calA$ are called here the \textit{antisymmetric functionals}, written as (possibly infinite) sequences: $T=(T_p)_{p\in\Na}$, where the components $T_p\in \calA^p=\calE'(\M^p,V^{\otimes p})^a$ are referred to as homogenous functionals. The evaluation of $T\in \calA$ on an element $\mathscr{C}(\M,V)\ni u=\bigoplus\limits_{p=0}^n u^{(p)}$, $u^{(p)}\in\calE(\M^p,V^{\otimes p})^a$ is understood as:
\begin{equation}
T(u)=\sum_{p=0}^n \langle T_p,u^{(p)}\rangle\equiv \sum_{p=0}^n T_p(u^{(p)})\,,
\end{equation}
where $\langle.,.\rangle$ denotes the natural duality between $\calE(\M^p,V^{\otimes p})^a$ and $\calE'(\M^p,V^{\otimes p})^a$.
Note that each $T\in \calA$ evaluated on an element of $\mathscr{C}(\M,V)$ is always a sum of finitely many terms. We can equip $\calA$ with the antisymmetric wedge product \cite{Scharf,Roep}:
\begin{equation}
S\wedge T(u_1,\ldots,u_{p+q})\doteq \frac{1}{p!q!}\sum\limits_{\pi\in P_{p+q}}(-1)^{\pi}S(u_{\pi(1)},\ldots,u_{\pi(p)})T(u_{\pi(p+1)},\ldots,u_{\pi(p+q)})\,, 
\end{equation}
for $S\in \calA^p(E)$, $T\in \calA^q(E)$, $u_i\in E$. The definition of  $\wedge$ can be now extended to $\mathscr{C}(\M,V)$ by continuity. 
 The wedge product is associative and antisymmetric for homogenous elements:
\begin{equation}
(R\wedge S)\wedge T=R\wedge(S\wedge T),\qquad S\wedge T=(-1)^{|S||T|}T\wedge S\,,
\end{equation}
where $|S|$, $|T|$ denote the grades of $S$, $T$ respectively. One can define a derivative on homogenous elements and extend it by linearity to the whole of $\calA$ with the prescription \cite{Roep}:
\begin{eqnarray}
d_h: \calA^p&\rightarrow& \calA^{p-1},\quad h\in E\nonumber\\
(d_h T)(u)&\doteq&T(h\wedge u),\quad T\in \calA^p,\ u\in\calE\left(\M^{p-1},V^{\otimes (p-1)}\right)^a,\ p>0\label{der}\\
d_h T&=&0\quad T\in \calA^0\,. \nonumber
\end{eqnarray}
It is easy to verify, that the ``derivative'' $d$ defined in (\ref{der}) has the following properties:
 \begin{enumerate}
 \item $d_h$ is a graded derivation for every $h\in E$ and $S,T$ homogenous, i.e.
 \begin{equation}
 d_h(\alpha S+\beta T)=\alpha d_hS+\beta d_hT,\qquad d_h(S\wedge T)=(d_hS)\wedge T+(-1)^{|S|}S\wedge d_h T\,,
 \end{equation}
 \item \label{dis1} for each $T\in \calA$ it induces a map:
 \begin{eqnarray}
T^{(1)}: \mathscr{C}(\M,V)&\rightarrow& L(E,\R)\label{d1}\,,\\
\left<T^{(1)}(u);h\right>&\doteq& d_hT(u)\,.
 \end{eqnarray}
 Moreover $T^{(1)}(u)$ is continuous (i.e. an element of $E'$) for every $u\in \mathscr{C}(\M,V)$. An object  $\left<T^{(1)}(.),h\right>\in\calA$ corresponds to the (formal) notion of ``the left variational derivative'' of $T$. The definition given here agrees with \cite{Roep,Berez,H}. One can define also the ``right derivative'' by a suitable modification of the definition of $d_u$. 
 \item property \ref{dis1} generalizes to:
 \begin{eqnarray}
T^{(k)}: \mathscr{C}(\M,V)&\rightarrow& L_{\alte}(\underbrace{E\times\ldots\times E}_k;\R),\quad T\in\calA^p,\  k<p\,,\\
\left<T^{(k)}(u);h_k,\ldots,h_1\right>&\doteq& d_{h_k}\ldots d_{h_1}T(u)=T(h_k\wedge\ldots\wedge h_1\wedge u)\,.
 \end{eqnarray} 
  Moreover $T^{(k)}(u)$ is jointly continuous for every $u\in \mathscr{C}(\M,V)$.
 \item is ``anticommutative'' in the following sense:
 \begin{equation}
 d_{h_1}d_{h_2}T=- d_{h_2}d_{h_1}T,\quad \forall T\in \calA
 \end{equation}
 \end{enumerate}
Analogously to the commutative case one can consider a particular class of elements of $\calA^1$, namely the evaluation functionals:
\begin{equation}
\calA^1\ni\Phi_x^a,\quad \Phi_x^a(u)\doteq u(x)^a,\quad\textrm{where }x\in\M,a=1,\ldots,k \textrm{ and }u\in E\,.
\end{equation}
Applying the wedge product to these functionals we get a relation:
\begin{equation}
\Phi_x^a\wedge\Phi_y^b=-\Phi_y^b\wedge\Phi_x^a\,.
\end{equation}
\subsection{Distributions}
With the kinematical structure introduced in the previous subsection we can start to proceed toward the dynamics. As indicated in the introduction, the functional approach relies heavily on the theory of distributions. Therefore the next step is to include the functional-analytic aspects into our framework. A natural formulation  in case of fermionic fields involves distributions with values in a graded algebra. Since it needs a certain generalization of the usual setting, we devote this section to introduce some abstract mathematical structures that are needed.

We start with some basic definitions concerning distributions. For details see: \cite{Hoer,Sch0,Sch1,Sch2}. To fix the notation, we define: $\calE(\M)\doteq\Ci(\M,\R)$, $\calD(\M)\doteq\Ci_0(\M,\R)$ and  $\mathscr{S}(\M)$ denotes the space of Schwartz functions. These function spaces are equipped with the topology of uniform convergence on compact subsets of $\M$. The corresponding distribution spaces are  defined as the topological duals. 
This can be also generalized to vector-valued distributions:
\begin{definition}\label{vector}
Let $X$ be a locally convex topological vector space with a topology defined by a separable family of seminorms $\{p_\alpha\}_{\alpha\in I}$. We say that $T$ is a distribution on $\M$ with values in $X$ if it is a continuous linear mapping from $\calD(\M)$ to $X$.
\end{definition}
The theory of vector-valued distributions was developed in \cite{Sch1,Sch2}.
 A slight modification was proposed in \cite{Kom}, where one uses sequential completeness instead of quasi-completeness. Most results from the theory of scalar-valued distributions generalize to this setting. The main difficulty lies in the fact, that there is no natural notion of tensor product for locally convex vector spaces. Also the approximation property does not hold in general. The situation is much easier if the space $X$ is nuclear and (sequentially) complete (see \ref{tops} and \cite{Jar} for details). This is the case if we take $\calA$ with the weak topology $\tau_\sigma$. A few details concerning the topologies are given in \ref{tops}. By the space of distributions with values in $\calA$ we shall understand  $\calD'(\M)_c\hat{\otimes}\calA$, where $\calD'(\M)_c$ denotes the dual of $\calD(\M)$ with the topology of uniform convergence on compact sets and $\hat{\otimes}$ is the sequential completion of the tensor product with respect to the tensor product topology\footnote{In general one has to distinguish between the projective and injective tensor product, but in case of nuclear vector spaces, these notions coincide. For detailed discussion see \cite{Jar,Sch1,Sch2,Kom}.}. 

The notion of vector-valued distributions enables us to formulate the classical theory of fermionic fields in a mathematically elegant way. Note that the map (\ref{d1}) can be treated as an element of $\calD'(\M,V)_c\hat{\otimes}\calA$, i.e. a distribution with values in a Grassman algebra. One can generalize all known operations like convolution, Fourier transform and pullback to such objects. We shall recall them here to set the notation and we refer to \cite{Hoer,Sch1,Sch2} for details. Since  $\calD'(\M,V)_c\hat{\otimes}\calA\cong\calD'(\M)_c\hat{\otimes}V\hat{\otimes}\calA$ and $V$ is finite dimensional we provide the definitions for the case $V=\R$, without the loss of generality\footnote{In general one needs an inner product structure on $V$. Since $V$ is finite dimensional, it can be always introduced with the natural pairing of $V$ and $V'\cong V$. In physical examples this dual pairing is usually provided by some natural structure. In case of Dirac fields, this is the pairing between spinors and cospinors and for the ghost fields in gauge theories it is induced by the Killing form of the gauge algebra.}.
\begin{definition}
Let  $T=t\otimes f$ and $\phi=\varphi\otimes g$, where $f,g\in \calA$, $t\in\calD'(\M)_c$ and $\varphi\in\calD(\M)$. We have an antysymmetric bilinear product on $\calA$ defined as: $m_a(T,S)\doteq T\wedge S$. We define the convolution of $T$ and $\phi$ by setting:
\begin{equation}
(T*_a\phi)(x)\doteq t(\ph(x-.))\otimes m_a(f,g)\,.
\end{equation}
The extension by the sequential continuity to $\calD'(\M)_c\hat{\otimes}\calA$ defines a convolution of a vector-valued distribution with a vector-valued function.
\end{definition}
\begin{definition}
Let  $T=t\otimes f$ and $S=s\otimes g$, where $f,g\in \calA$, $t\in\calE'(\M^2)_c$ and $s\in\calD'(\M)_c$. We define the convolution of $T$ and $S$ by setting:
\begin{equation}
T*_aS\doteq \int t(.,y)s(y)dy\otimes m_a(f,g)\,,
\end{equation}
This expression is well defined by \cite[4.2.2]{Hoer} and can be extended by the sequential continuity to an arbitrary $S\in\calD'(\M)_c\hat{\otimes}\calA$,  $T\in\calE'(\M)_c\hat{\otimes}\calA$.
\end{definition}
\begin{definition}
In a similar spirit we define the evaluation of $T=t\otimes f$ on $\phi=\varphi\otimes g$,  by:
\begin{equation}
\left<T,\phi\right>^a\doteq \left<t,\varphi\right>\otimes m_a(f,g)\,,
\end{equation}
where $f,g\in \calA$, $t\in\calD'(\M)_c$ and $\varphi\in\calD(\M)$. Also this can be extended by the sequential continuity to  $\calD'(\M)_c\hat{\otimes}\calA$.
\end{definition}
\begin{definition}
Let $T\in\mathscr{S}'(\M)_c\hat{\otimes}\calA$. We define $\hat{T}\in\mathscr{S}'(\M)_c\hat{\otimes}\calA$, the Fourier transform of $T$ as:
\begin{equation}
\hat{T}(\phi)=T(\hat\phi)\qquad\phi\in\mathscr{S}(\M)\,.
\end{equation}
\end{definition}
Also the notion of the wave front set \cite{Hoer} can be extended to distributions with values in a lcvs. The case of Banach spaces was already treated in detail in \cite{Ko}.
\begin{definition}
Let $\{p_\alpha\}_{\alpha\in A}$ be the family of seminorms generating the locally convex topology on $\calA$. Let  $T\in\mathscr{S}'_c(\M)\hat{\otimes}\calA$. A point $(x,\xi_0)\in T^*\R^{n}\setminus 0$ is not
in $\textrm{WF}(T)$, if and only if $p_\alpha(\widehat{\phi T}(\xi))$
is fast decreasing as $|\xi|\rightarrow\infty$ for all $\xi$ in an open
conical neighbourhood of $\xi_0$, for some $\phi\in \calD(\M)$
with $\phi(x)\neq 0$, $\forall \alpha\in A$.
\end{definition}
With the notion of the wave front set we can define a ``pointwise product'' of two distributions $T,S\in\calD'(\M)_c\hat{\otimes}\calA$ by a straightforward extension of \cite[8.2.10]{Hoer}:
\begin{proposition}
Let $T,S\in\calD'(U)_c\hat{\otimes}\calA$, $U\subset\M$ (open). The product $T\cdot_aS$ can be defined as the pullback of $m_a\circ(T\otimes S)$ by the diagonal map $\delta:U\rightarrow U\times U$ unless $(x,\xi)\in\textrm{WF}(T)$ and $(x,\xi)\in\textrm{WF}(S)$ for some $(x,\xi)$.
\end{proposition}
Obviously we have:  $T\cdot_aS=(-1)^{|S||T|}S\cdot_aT$, whenever these expressions are well defined. In the following we shall also use a more suggestive notation: $T\cdot_aS\doteq\left<T,S\right>^a$.

Now we want to impose some regularity conditions on the distributions we want to consider. This is important for the definition of the subspace of $\calA$, where the Peierls bracket is well defined.
Let $\Xi_n\doteq\{(x_1,...,x_n,k_1,...k_n)| (k_1,...k_n)\notin (\overline{V}_+^n \cup \overline{V}^n_-)\}$ be an open cone. We denote by $\calF^n\doteq\calE'_{\Xi_n}(\M^n,V^{\otimes n})$ the subspace of $\calA^n=\calE'(\M^n,V^{\otimes n})$ consisting of distributions with wave front set contained in $\Xi_n$. We want to restrict to antisymmetric functionals that are elements of: $\calF\doteq\prod\limits_{n=0}^\infty \calF^n$. With analogy to the bosonic case we call these objects \textit{microcausal functionals}. Space $\calF$ has two important subspaces: $\calF^+\doteq\prod\limits_{p=0}^\infty \calF^{\,2p}$ and $\calF^-\doteq\prod\limits_{p=0}^\infty \calF^{\,2p+1}$. Elements of $\calF^+$ would be called even and elements of $\calF^-$, odd. Since we are using the weak topology on $\calA$, $\calF$ can be equivalently characterized as a space of functionals satisfying:
\begin{equation}
\textrm{WF}(F^{(n)}(u)) \cap (\M^n \times (\overline{V}_+^n \cup \overline{V}^n_-)) = \emptyset\qquad \forall u\in\euC(\M,V),\ n\in\Na\,,\label{WFcond}
\end{equation}
where $F^{(n)}(u)$ is treated as an element of $\calE'(\M^n,V^{\otimes n})_c\widehat{\otimes}\calA$.
We can impose a more restrictive condition on the wave front set to define a subspace $\calF_{loc}^p\subset\calF^p$ consisting of \textit{local functionals}. A functional $F\in\calF^p$ is called local if it satisfies:
\begin{equation}
\textrm{WF}(F)\perp T\Delta^p(\M)\,,
\end{equation}
and is supported on $\Delta^p(\M)$,
where $\Delta^p(\M)\doteq\left\{(x,\ldots,x)\in \M^p:x\in \M\right\}$ is the thin diagonal of $\M^p$.
This more abstract notion agrees with the notion of locality (\ref{local}) mentioned in the introduction.
 $\calF$ can be equipped with various topologies, for example the weak topology $\tau_\sigma$ inherited from $\calA$. Since we want to have control of the wavefront sets of functional derivatives, we shall use instead the so called H\"ormander topology \cite{Hoer}. 
 Let $\Gamma_n\subset \Xi_n$ be a closed cone contained in $\Xi_n$. We introduce (following \cite{Hoer,Bae,BDF}) the family of seminorms on $\calE'_{\Gamma_n}(\M^n,V^{\otimes n})$ given by: $p_{n,\phi,C,k} (T) = \sup_{\xi\in C}\{(1 + |\xi|)^k |\widehat{\phi T}(\xi)|\}$, where the index set consists of $(n,\phi,C,k)$ such that $k\in \Na_0$, $\phi\in \calD(U)$ and $C$ is a closed cone in $\R^n$ with $(\supp ( \phi ) \times C) \cap\Gamma_n = \emptyset$. These seminorms, together with the seminorms of the weak topology provide a defining system for a locally convex topology denoted by $\tau_{\Gamma_n}$. To control the wavefront set properties inside open cones, we take an inductive limit. It is easy to see that, to form this inductive limit, one can choose the family of closed cones contained inside $\Xi_n$ to be countable. The resulting topology will be denoted by $\tau_{\Xi_n}$. The space $\calF$ can be now equipped with the direct product topology denoted by $\tau_{\Xi}$. Since each of the topologies $\tau_{\Gamma_n}$ is nuclear \cite{BDF} so $\tau_{\Xi}$ is nuclear as well. It has also the property of sequential completeness, so from now on we shall always take $\calF$ with this topology, unless stated differently. For $F\in\calF^n=\calE'_{\Xi_n}(\M^n,V^{\otimes n})$ we shall often use the notation:
\begin{equation}
F(u)=\sum\limits_{a_1=1}^k\ldots\sum\limits_{a_p=1}^k\int dx_1\ldots dx_p\, F_{a_1\ldots a_p}(x_1,\ldots,x_p) u^{a_1\ldots,a_p}(x_1,\ldots,x_p)\,,\label{Fp0}
\end{equation} 
where $F_{a_1\ldots,a_p}(x_1,\ldots,x_p) $ is an integral kernel of a compactly supported distribution with the wavefront set contained in $\Xi_p$. Since $\Phi_{x_1}^{a_1}\wedge\ldots\wedge\Phi_{x_p}^{a_p}(u)=u^{a_1\ldots,a_p}(x_1,\ldots,x_p)$ we can write (\ref{Fp0}) formally and analogously to the bosonic case \cite{DF02,DF04}, as:
\begin{equation}
F(u)=\sum\limits_{a_1=1}^k\ldots\sum\limits_{a_p=1}^k\int dx_1\ldots dx_p\, F_{a_1\ldots,a_p}(x_1,\ldots,x_p) \Phi_{x_1}^{a_1}\wedge\ldots\wedge\Phi_{x_p}^{a_p}(u)\,,\label{Fp}
\end{equation}
In the following we shall suppress the vector space indices of the $\Phi$'s whenever possible. With this notation (\ref{Fp}) shall be written as: $F(u)=\int dx_1\ldots dx_p\, F(x_1,\ldots,x_p) \Phi_{x_1}\wedge\ldots\wedge\Phi_{x_p}(u)$.
%
\section{Dynamical structure}\label{Dyn}
%
\subsection{Equations of motion}\label{eqom}
%
A generalized lagrangian \cite{BDF} $S$ is defined to be a map $S:\calD(\M)\rightarrow\calF_\loc$ such that
\begin{eqnarray}
1.&&\supp(S(f))\subseteq \supp(f)\label{lagsupp}\,,\\
2.&&S(f\!+\!g\!+\!h)=S(f+g)\!-\!S(g)\!+\!S(g+h),\ \textrm{if}\ \supp(f)\cap\supp(h)=\emptyset\,.\label{lagadd}
\end{eqnarray}
The second condition is called \textit{additivity}. One can think of it as a weaker replacement for the notion of linearity. It is crucial in the quantum field theory, where one uses the test function as a localized coupling constant \cite{BDF}. To admit interaction terms that  
are of higher order in the coupling (for example $\sim f^2$) one has to drop the linearity condition, whereas the additivity still holds. 

The ``variation'' of $S(f)$ is understood in the sense of (\ref{der}) namely we require that: $ \left<S(f)^{(1)}(u),h\right>=0$ for all $h\in\calD(\M,V)$ and $f\in\calD(\M)$ such that $f\equiv 1$ on $K\doteq\supp\, h$. For this choice of $f$ we use notation:
\begin{equation}
\left<S(1)^{(1)}(u),h\right>=0\,.\label{eom}
 \end{equation}
Since we can choose $K$ arbitrary large, it follows that the equations of motion must hold true on the whole $\M$. If $S(f)\in\calF_\loc^2$, we can interpret (\ref{eom}) as a system of partial differential equations for $u\in E=\calE(\M,V)$. This is analogous to the bosonic case and allows us to define a solution space $E_S\subset E$. For higher order interactions this concept has to be modified. If the nonlinearity is present, the equation (\ref{eom}) contains Grassman-valued objects. Hence it cannot be seen as an equation for $C$-valued functions. Therefore it is more convenient to work purely on the algebraic level.  We define an ideal $\mathcal{J}_S\subset\calF$ as the one generated (in the algebraic and topological sense) by the set $\left\{\left<S(1)^{(1)}(.),h\right>\right\}_{h\in \calD(\M,V)}\subset\calF$ (equations of motion). Then we can construct the quotient $\calF/\mathcal{J}_S$ and write $S(f)$  in terms of equivalence classes $[\Phi_x]\in\calF/\mathcal{J}_S$.

One can linearize a given lagrangian $ S(f)$ in terms of the second derivative. 
We always assume that the lagrangian is even i.e. $S(f)^{(2)}$ is an element of $\calE'(\M^2,V^{\otimes 2})^a\hat{\otimes}\calF^+$. In some cases it can be inverted, so there exists $\Delta^*\in \calD'(\M^2,V^{\otimes 2})^a\hat{\otimes}\calF^+$ such that:
\begin{equation}
\left<S(1)^{(2)}*_a\Delta^*;h_1,.\right>=\delta*_ah_1\,.
\end{equation}
Formally this can be written as:
\begin{equation}
S(f)^{(2)}*_a \Delta^*=\delta*_a\1\label{delt}\,.
\end{equation}
By retarded and advanced Green's functions $\Delta^{R/A}$ we mean solutions of (\ref{delt}) satisfying in addition:
\begin{eqnarray}
\supp(\Delta^R)&\subset&\{(x,y)\in \M^2| y\in J^-(x)\}\label{retarded}\,,\\
\supp(\Delta^A)&\subset&\{(x,y)\in \M^2| y\in J^+(x)\}\label{advanced}\,.
\end{eqnarray}
Particularly, if $S(f)^{(2)}\in\calE'(\M^2,V^{\otimes 2})^a$ and it is a strictly hyperbolic operator, it can be shown that it has retarded and advanced Green's functions $\Delta^R$, $\Delta^A$.
\subsection{M{\o}ller maps}%
In section \ref{eqom} we defined the ``on-shell'' algebra as the quotient  $\calF/\calJ_S$, where $\calJ_S$ is the ideal generated by the equations of motion.  Analogously to the classical theory of bosonic fields we would like to compare theories with different actions $S_1$, $S_2$. This can be done most conveniently at the algebraic level. One can construct analogs of off-shell M{\o}ller maps \cite{BreDue}, $r_{S_1,S_2}:\calF\rightarrow\calF$, which intertwines the corresponding ideals $\calJ_{S_1}$ and $\calJ_{S_2}$. We require $r_{S_1,S_2}$ to have the following properties:
\begin{enumerate}
\item if $G\in\calJ_{S_1}$, then $r_{S_1,S_2}G\in\calJ_{S_2}$ (the intertwining property)\,,\label{id1}
\item $r_{S_2,S_3}\circ r_{S_1,S_2}=r_{S_1,S_3}$\,,\label{comp}
\item $(r_{S_1,S_2}G)(u)=G(u)$, $u\in\euC(\M,V)$ if $\supp(u)\cap(\supp(S_1-S_2)+\overline{V}_+)=\emptyset$\,,\label{suppr}
\item $G\mapsto r_{S_1,S_2}(G)$ is a homomorphism of $\calF$\,.\label{hom}
 \setcounter{saveenum}{\value{enumi}}
\end{enumerate}
Now let $S_1=S+\lambda F$ and $S_2=S$ for $S,F\in \calF_\loc^+$, $|S|=2$.
We assume that for $S^{(2)}(x,y)$ there exist retarded and advanced Green's functions such that $\Delta^{R}(x,y)=-(\Delta^{A}(y,x))^T$. We want to construct $r_{S+\lambda F,S}$ as a series in $\lambda$. 
Assume first that $r_{S+\lambda F,S}$ exists and for fixed $G$ we have a smooth (in the sense of calculus on locally convex vector spaces) map $r_{.,S}(G):\calF_\loc^+\rightarrow\calF_\loc$. Then we can use the generalized Taylor series expansion to obtain:
\begin{equation}
r_{S+\lambda F,S}(G)=\sum\limits_{k=0}^\infty\frac{\lambda^k}{k!}(d^k\,r_{.,S}(G))(S)[F^{\otimes k}]\doteq\sum\limits_{k=0}^\infty\frac{\lambda^k}{k!}R_{S,k}(F^{\otimes k},G)\,.\label{TS}
\end{equation}
To construct $r_{S+\lambda F,S}$ we shall try to reverse this reasoning and define $r_{S+\lambda F,S}$ by its power series. Each term should be a $(k+1)$-linear map, symmetric in the first $k$ arguments.
We require that the $0$-th order term is the identity map, and the first order term is the retarded product of $F,G\in\calF_\loc$, that is
\begin{equation}
R_{S,1}(F,G)=R_S(F,G)=\frac{d}{d\lambda}\Big|_{\lambda=0}r_{S+\lambda F,S}(G)\,.
\end{equation}
After \cite{DF02} we call higher order terms:  \textit{higher order retarded products}. From the conditions on $r_{S+\lambda F,S}$, we deduce those, which we want to impose on $R_{S,k}(F^{\otimes k},.)$. To fulfill \ref{id1} we postulate that:
\begin{equation}
r_{S+\lambda F,S}\left(\left<S^{(1)}+\lambda F^{(1)},h\right>\right)=\left<S^{(1)},h\right>\,,\label{id2}
\end{equation}
where $h\in\calD(\M,V)$ and $\left<S^{(1)}+\lambda F^{(1)},h\right>\in\calJ_{S+\lambda F}$. The fact that $\calJ_{S+\lambda F}$ is generated by elements of this form, together with condition \ref{hom} already suffices to fulfill \ref{id1}. From (\ref{id2}) follows a recursive condition on the retarded products:
\begin{equation}
R_{S,k}\left(F^{\otimes k}, \left<S^{(1)},h\right>\right)=-kR_{S,k-1}\left(F^{\otimes (k-1)},\left<F^{(1)},h\right>\right),\quad k>0\label{recur1}
\end{equation}
Particularly, for $k=1$ we have:
\begin{equation}
R_S\left(F, \left<S^{(1)},h\right>\right)=-\left<F^{(1)},h\right>\,.\label{ret0}
\end{equation}
There is still a big freedom in defining $R_S$. Particularly, we can use the analogy with bosonic fields and define it first for $F\in\calF_\loc^2$. In this case the equations of motion can be interpreted in terms of a dynamical system and one can define the corresponding M{\o}ller maps on the configuration space $E$. Let $E_{S}$ and  $E_{S+\lambda F}$ be the solution spaces corresponding to actions $S$ and $S+\lambda F$. It was already shown  in \cite{DF} that the map $\tilde{r}_{S+\lambda F,S}:E_{S}\rightarrow E_{S+\lambda F}$ can be constructed perturbatively. It was proven for the scalar field but it can be easily generalized to the fermionic case for $|F|=2$. The existence of nonperturbative solutions will be discussed in \cite{Pedro}. Therefore we postulate two more conditions on $r_{S+\lambda F,S}$:
\begin{enumerate}
\setcounter{enumi}{\value{saveenum}}
\item If $|F|=|S|=2$ then: $r_{S+\lambda F,S}(\calF^1)\subseteq\calF^1$\,,\label{RF2}
\item $r_{S+\lambda F,S}(G)(u)=G\circ\tilde{r}_{S+\lambda F,S}(u),\qquad u\in E,\quad F,S\in\calF^2,\ G\in\calF^1$\,.
 \setcounter{saveenum}{\value{enumi}}\label{RFtilde}
\end{enumerate}
From conditions \ref{RF2}, \ref{RFtilde} and proposition 1 of \cite{DF} it is clear that for $F,S\in\calF^2_\loc$ and $G\in\calF^1_\loc$ the retarded product can be expressed as:
\begin{equation}
R_S(F,G)=-\left<G^{(1)},\Delta^R*_aF^{(1)}\right>^a\,.\label{ret1}
\end{equation}
To extend this definition to higher order functionals, we require the (graded) Leibniz rule in the left and right argument:
\begin{enumerate}
\setcounter{enumi}{\value{saveenum}}
\item $R_S(F_1\wedge F_2,G)=F_1\wedge R_S(F_2,G)+R_S(F_1,G)\wedge F_2$\,,
\item $R_S(F,G_1\wedge G_2)=G_1\wedge R_S(F,G_2)+R_S(F,G_1)\wedge G_2$\,.
\end{enumerate}
By continuity we can now extend $R_S$ to general $F\in\calF_\loc^+$ and $G\in\calF_\loc$. It is given by the same formula, namely (\ref{ret1}). One can check with an explicit calculation, that (\ref{ret0}) is fulfilled. Now we proceed like in \cite{DF02}. Condition \ref{comp} implies that:
\begin{equation}
r_{S+\lambda F_1,S}\circ r_{S+\lambda F_1+\mu F_2,S+\lambda F_1}=r_{S+\lambda F_1+\mu F_2,S}\,.
\end{equation}
The comparison of the coefficients with respect to the powers of $\mu$ yields:
\begin{equation}
r_{S+\lambda F_1,S}(R_{S+\lambda F_1}(F_2,G))=\frac{d}{d\mu}\Big|_{\mu=0}r_{S+\lambda F_1+\mu F_2,S}(G)\,.
\end{equation}
Again we can compare the coefficients with respect to the powers in $\lambda$ and obtain the following recursive condition:
\begin{equation}
  R_{S,n+1}(F_1^{\otimes n}\otimes F_2,G)=-\sum_{l=0}^n \binom{n}{l}
R_{S,l}\Bigl( F_1^{\otimes l},(-1)^{|F_2|+1}\left<F_{2}^{(1)}, 
        \Delta_{S+\lambda F_1}^{A\,(n-l)}*_a
       G^{(1)}\right>^a\Bigr),\label{nreta}
\end{equation}
where
\begin{gather}
  \Delta_{S+\lambda F_1}^{A\,(k)}\doteq
\frac{d^k}{d\lambda^k}\Big|_{\lambda =0}
\Delta_{S+\lambda F_1}^{A}=(-1)^k k!\, \Delta_{S}^A*_a
 F_{1}^{(2)}*_a\Delta_S^{A}*_a\ldots*_aF_{1}^{(2)}*_a\Delta_S^{A}\>\,.\label{nretb}\end{gather}
The formula for $ \Delta_{S+\lambda F_1}^{A\,(k)}$ is the graded counterpart of the relation (42) in \cite{DF02}. Particularly we can set $F_1=F_2=F$ and obtain:
\begin{equation}
 R_{S,n+1}(F^{\otimes (n+1)},G)=-\sum_{l=0}^n \binom{n}{l}
R_{S,l}\Bigl( F^{\otimes l},\left<F^{(1)}(x),
        \Delta_{S+\lambda F}^{A\,(n-l)}*_a
        G^{(1)}\right>^a\Bigr)\,.\label{nretc}
\end{equation}
Now we can prove that the Taylor series (\ref{TS}) with coefficients given by (\ref{nretc}) defines a map $r_{S+\lambda F,S}:\calJ_{S+\lambda F,S}\rightarrow\calJ_{S}$. First we have to check if the recursive condition (\ref{recur1}) is fulfilled. This is the case since we have:
\begin{multline}
R_{S,k+1}\left(F^{\otimes k+1}, \left<S^{(1)},h\right>^a\right)=-\sum_{l=0}^{k} \binom{k}{l}
R_{S,l}\Bigl( F^{\otimes l},\left<F^{(1)},\left<\Delta_{S+\lambda F}^{A\,(k-l)}*_aS^{(2)};.,h\right>^a\right>^a\Bigr)=
\nonumber\\
        =k\sum_{l=0}^{k-1} \binom{k-1}{l}
R_{S,l}\Bigl( F^{\otimes l},\left<F^{(1)},\left<\Delta_{S+\lambda F}^{A\,(k-l-1)}*_aF^{(2)};.,h\right>^a\right>^a
        \Bigr)-R_{S,k}(F^{\otimes k},\left<F^{(1)},h\right>^a)=\nonumber\\
        =-(k+1)R_{S,k}(F^{\otimes k},\left<F^{(1)},h\right>^a)\,.
         \end{multline}
Next we show that $G\mapsto r_{S+\lambda F,S}(G)$ is a homomorphism of $\calF$. 
\begin{proposition}
 Let $S,F\in\calF^+$ and let $r$ be defined  by the Taylor series (\ref{TS}) with the coefficients given by (\ref{nretc}). Then it holds (in every order in $\lambda$):
\begin{equation} 
r_{S+\lambda F,S}(G\wedge H)=r_{S+\lambda F,S}(G)\wedge r_{S+\lambda F,S}(H)\,.\label{compo}
\end{equation}
\end{proposition}
\begin{proof}
First we show the identity:
\begin{equation}
R_{S,n}(F^{\otimes n}, G\wedge H)=\sum\limits_{k=0}^n\left(n\atop k\right)R_{S,k}(F^{\otimes k}, G)\wedge R_{S,n-k}(F^{\otimes (n-k)}, H)\,.\label{ind0}
\end{equation}
This can be proved by induction. For $n=1$ we have: $l.h.s.=G\wedge H=r.h.s.$. Now we assume that (\ref{ind0}) is satisfied at the order $n$ and prove the induction step:
\begin{equation}
R_{S,n+1}(F^{\otimes (n+1)}, G\wedge H)=-\sum_{l=0}^n \binom{n}{l}
R_{S,l}\Bigl( F^{\otimes l},\left<F^{(1)},
        \Delta_{S+\lambda F}^{A\,(n-l)}*_a (G\wedge H)^{(1)}\right>^a\Bigr)\,.\label{GH}
\end{equation} 
First we apply the graded Leibniz rule. With the use of the induction hypothesis and after changing the order of summation and renaming the indices the first term of (\ref{GH}) can be written as:
\begin{multline}
-\sum_{l=0}^n \binom{n}{l}
R_{S,l}\left( F^{\otimes l},\left< F^{(1)},\Delta_{S+\lambda F}^{A\,(n-l)}*_a
       G^{(1)}\right>^a\wedge H\right)=\\
=-\sum_{k=0}^n\sum\limits_{l=k}^n \binom{n}{l}\binom{l}{k}
R_{S,l-k}\left( F^{\otimes (l-k)},\left<F^{(1)}, 
        \Delta_{S+\lambda F}^{A\,(n-l)}*_a
       G^{(1)}\right>^a\right)\wedge R_{S,k}(F^{\otimes k},H)=\\
=-\sum\limits_{k=0}^n \sum\limits_{l=0}^k \binom{n}{k}\binom{k}{l}R_{S,l}\left(F^{\otimes l},\left<F^{(1)},
        \Delta_{S+\lambda F}^{A\,(k-l)}*_a
       G^{(1)}\right>^a\right)\wedge R_{S,n-k}( F^{\otimes (n-k)},H)=\\
       =-\sum\limits_{k=0}^n \binom{n}{k} R_{S,k+1}( F^{\otimes (k+1)},G)\wedge R_{S,n-k}( F^{\otimes (n-k)},H)\,.
\end{multline}
The last equality is a consequence of the definition of higher order retarded products. It follows now that:
\begin{multline}
R_{S,n+1}(F^{\otimes (n+1)}, G\wedge H)=\sum\limits_{k=0}^n\binom{n}{k}R_{S,k+1}(F^{\otimes(k+1)},G)\wedge R_{S,n-k}(F^{\otimes(n-k)},H)+\\
+(-1)^{|H||G|}
\sum\limits_{k=0}^n\binom{n}{k}R_{S,k+1}(F^{\otimes(k+1)},H)\wedge R_{S,n-k}(F^{\otimes(n-k)},G)=\\
=\sum\limits_{k=0}^{n+1}\binom{n+1}{k}R_{S,k}(F^{\otimes k},G)\wedge R_{S,n-k+1}(F^{\otimes(n-k+1)},H)
\end{multline}
This proves the induction step. From the induction principle it follows that (\ref{ind0}) holds. Yet (\ref{ind0}) is simply the Taylor expansion of (\ref{compo}), so (\ref{compo}) holds in the sense of formal power series.
\end{proof}
This proves condition \ref{hom}. Together with (\ref{recur1}) this implies that condition \ref{id1} is fulfilled as well. Condition \ref{suppr} is fulfilled because of support properties of $\Delta^R$ and \ref{comp} holds from the definition. It remains to show that the series (\ref{TS})  defines indeed an element of $\calF$. To do it we have to check if in each degree we obtain a finite expression. First we assume that $F$ has only terms of degree higher than $2$. Then we get convergence in each grade, since $R_{S,n}(F^{\otimes n},G)$ has a degree that increases with $n$, i.e. $|R_{S,n}(F^{\otimes n},G)|=|G|+n(|F|-2)>|R_{S,n-1}(F^{\otimes (n-1)},G)|$ and therefore the sum $(r_{S+\lambda F,S}(G))(u)=\sum\limits_{k=0}^\infty\frac{\lambda^k}{k!}(R_{S,k}(F^{\otimes k},G))(u)$, $u\in\euC(\M,V)$ has only finitely many non-vanishing terms. For $|F|=2$ we cannot use this argument, but we can instead construct $r_{S+\lambda F,S}$ with the use of $\tilde{r}_{S+\lambda F,S}$. The existence of $\tilde{r}$ has to be showed with the same argument as for the bosonic case \cite{Pedro}.

As a final remark, we discuss the existence of an inverse mapping: $r_{S+\lambda F,S}^{-1}:\calJ_S\rightarrow\calJ_{S+\lambda F,S}$ for $|F|>2$. This map can always be defined as a formal power series, since the first term in the expansion (\ref{TS}) is the identity. Moreover, as discussed above, this provides a well defined element of $\calF$.

The retarded product given by formula (\ref{ret1}) can be extended in the left argument from $\calF^+_\loc$ to $\calF_\loc$ by postulating (for homogenous elements):
\begin{equation}
R_S(F,G)=(-1)^{|F|+1}\left<F^{(1)},\Delta^R_{S}*_aG^{(1)}\right>^a\,.\label{ret2}
\end{equation}
We can extend this definition to the whole $\calF_\loc$ by linearity and continuity. The additional sign factor is necessary since we use only left derivatives instead of right and left ones. Analogously, the advanced product is defined as:
\begin{equation}
A_S(F,G)=(-1)^{|F|+1}\left<F^{(1)},\Delta^A_{S}*_aG^{(1)}\right>^a\,.\label{adv2}
\end{equation}
\subsection{Peierls bracket}
Let $\Delta(x,y)=\Delta^R(x,y)-\Delta^A(x,y)=( \Delta(y,x))^T$. For $F\in\calF_\loc^{\,p}$, $G\in\calF_\loc^{\,q}$ we can define a  Poisson structure by a definition analogous to the bosonic case:
\begin{equation}
\{F,G\}_S=R_S(F,G)-A_S(F,G)\,,
\end{equation}
where the retarded and advanced products are given by (\ref{ret2}), (\ref{adv2}). This can be written as:
\begin{equation}
\{F,G\}_S=(-1)^{|F|+1}\left<F^{(1)},\Delta_{S}*_aG^{(1)}\right>^a\,.\label{peierls2}
\end{equation}
We can extend this definition also to $F\in\calF^p$, $G\in\calF^q$.
By \cite[Thm. 8.2.10]{Hoer} the pointwise product of distributions appearing in 
(\ref{peierls2}) is well defined. 
Since $S^{(2)}$ is assumed to be even, so is $\Delta$. Therefore, for homogenous elements, we have the graded anticommutativity:
\begin{equation}
\{F,G\}_S=-(-1)^{|F||G|}\left<G^{(1)},\Delta_{S}*_aF^{(1)}\right>^a=-(-1)^{|F||G|}\{G,F\}_S\,.
\end{equation}
By modifying slightly the proof of the Jacobi identity for the bosonic fields given in \cite{Jac} one can show that for homogenous elements  $F,G,H\in\calF$:
\begin{equation}
\{\{F,G\}_S,H\}_S(-1)^{|F||H|}+\{\{G,H\}_S,F\}_S(-1)^{|F||G|}+\{\{H,F\}_S,G\}_S(-1)^{|G||H|}=0\,.
\end{equation}
Also the graded derivation law is fulfilled, namely
\begin{equation}
\{F\wedge G, H\}_S=(-1)^{|G||H|}\{F, H\}_S\wedge G+F\wedge\{G, H\}_S\,.
\end{equation}
In section \ref{eqom} we defined the ideal of $(\calF,\wedge)$ generated by equations of motion for a given action functional $S$. We denoted this ideal as $\calJ_S$. Now we prove that this is also a Poisson ideal with respect to the $\{.,.\}_S$ structure:
\begin{proposition}
Let $\calJ_S$ be the $(\calF,\wedge)$-ideal generated by elements of the form $\left<S^{(1)},h\right>$. Then $\calJ_S$ is a Poisson ideal of the Poisson algebra $(\calF,\{.,.\}_S)$.
\end{proposition}
\begin{proof}
A general element of $\calJ_S$ can be written as a limit of a sequence of elements of the form:  $F\wedge~\left<S^{(1)},h\right>$ for $F\in\calF$, $h\in \calD(\M,V)$. Inserting this in formula (\ref{peierls2}) yields:
\begin{equation}
\{F\wedge\left<S^{(1)},h\right>,G\}_S=(-1)^{|F|+|S|}\left<\left(F\wedge\left<S^{(1)},h\right>\right)^{(1)},\Delta*_a G^{(1)}\right>^a\,.
\end{equation}
With the use of (\ref{delt}) and the definition of the causal propagator we obtain:
\begin{eqnarray}
\{F\wedge\left<S^{(1)},h\right>,G\}_S&=&(-1)^{|F|}\left(\left<F^{(1)},\Delta*_a G^{(1)}\right>^a\right)\wedge\left<S^{(1)},h\right>=\\
&=&-(\{F,G\}_S)\wedge\left<S^{(1)},h\right>\in\calJ_S\,.\nonumber
\end{eqnarray}
\end{proof}
The above proposition shows that we can now take the quotient of $\calF$ by $\calJ_S$ and we obtain the Poisson algebra of observables: $(\calF_S,\{.,.\}_S)\doteq(\calF/\calJ_S,\{.,.\}_S)$. Different action functionals define different Poisson structures on $\calF$. It turns out, that (like in the bosonic case \cite{DF}) those structures are intertwined by M{\o}ller maps.
\begin{proposition}
The retarded (advanced) M{\o}ller maps are canonical transformations for Poisson structures induced by action functionals. Namely:
\begin{equation}
\{r_{S_2,S_1}(F),r_{S_2,S_1}(G)\}_{S_1}=r_{S_2,S_1}(\{F,G\}_{S_2})\,.\label{cantr}
\end{equation}
\end{proposition}
\begin{proof}
The proof is analogous to the bosonic case \cite{DF}. The infinitesimal version of (\ref{cantr}) is simply:
\begin{equation}
\{R_S (H, F ), G\} + \{F , R_S (H, G)\} = 
R_S (H, \{F, G\}) + \frac{d}{ d\lambda}\Big|_ {\lambda=0} (R_{S +\lambda H} (F, G) - A_{S +\lambda H} (F, G))\,.
\end{equation}
This in turn can be verified by a straightforward calculation, using the fact that:
\begin{eqnarray}
 \left<(\Delta_S^{A})^{(1)},h\right>&=&
    -\Delta_S^{A}*_a
    \left<S^{(3)};h,.,.\right>*_a
    \Delta_S^{A}\,,\\
     \frac{d}{d\lambda}\Big|_{\lambda=0}
     \Delta_{S+\lambda H}^{A}&=&
     -\Delta_S^{A}*_a H^{(2)}*_a \Delta_S^{A}.
\end{eqnarray}
The second identity can be proved analogously as the formula (37) of \cite{DF02}.
\end{proof}
\section{Gross-Neveu model}\label{grossn}
After introducing the general formalism, we shall now apply it in a concrete example. We want to construct the algebra of classical observables for the Gross-Neveu model. Since we shall do it perturbatively, we start with the free action, namely we consider a free Dirac field in Minkowski spacetime. Let $D\M$ ($D^*\M$) be the spinor (cospinor) bundle. We take the Whitney sum $D\M\oplus D^*\M$ and define the configuration space to be $E=\mathscr{E}(D\M\oplus D^*\M)$, the set of smooth sections. Let $E\ni\tilde{u}=u\oplus \overline{u}$. To be consistent with the standard approach we introduce a following notation for evaluation functionals:
\begin{eqnarray}
\Psi_{x\ A}(\tilde{u})&=&u_{A}(x)\,,\\
\overline{\Psi}_x^{\dot B}(\tilde{u})&=&\overline{u}^{\dot B}(x)\,.
\end{eqnarray}
The generalized action functional for the free Dirac field takes the form:
 \begin{equation}
S_0(f)(u)=\int dx\, f(x)(\overline{\Psi}_x\wedge(i\partial\!\!\!/-m)\Psi_x)(u)\quad f\in\calD(\M)\,,\label{dirac}
 \end{equation}
 where $u\in\euC(\M,E)$ and the derivative $\partial_\mu$ is a weak derivative, i.e.: $(\partial_\mu \Psi_x)(u)\doteq \Psi_x(\partial_\mu u)=(\partial_\mu u)(x)$. The equations of motion take the form:
 \begin{equation}
 \left<S_0(1)^{(1)}(\tilde{u}),\tilde{h}\right>=\int dx\, (\overline{\Psi}_x\wedge(i\partial\!\!\!/-m)\Psi_x)(\tilde{h}\wedge\tilde{u})
= \int dx\, \left<\tilde{h}(x), D\Psi_x\oplus -D^*\overline{\Psi}_x\right>_E(\tilde{u})\,,
 \end{equation}
where $D\Psi_x\doteq(i\partial\!\!\!/-m)\Psi_x$, $D^*\overline{\Psi}_x\doteq-\overline{\Psi}_x(i\overleftarrow{\partial\!\!\!/}+m)$ and $<.,.>_E$ denotes the dual pairing on $E$ induced by the pairing between spinors and cospinors.
Let $\calJ_0$ be an ideal generated by the elements $D\Psi_x\oplus -D^*\overline{\Psi}_x$. The on-shell algebra of functionals is defined as $\calF/\calJ_0$.
The second derivative of $S_0(f)$ takes the form:
\begin{equation}
\left<S_0(1)^{(2)};\tilde{h}_1,\tilde{h}_2\right>=
\int dx\, \left(\overline{h_1}(i\partial\!\!\!/-m)h_2-\overline{h_2}(i\partial\!\!\!/-m)h_1\right)\,.
\end{equation}
It is convenient to write $S_0(1)^{(2)}(x,y)$ as a block matrix in the basis $(u_A,\overline{u}^{\dot A})$:
\begin{equation}
S_0(1)^{(2)}(x,y)=\delta(x-y)\left(\begin{array}{cc}
0&D^{*^T}(x)\\
-D(x)&0
\end{array}\right)\,,
\end{equation}
where ${}^T$ denotes the transpose of a matrix. We can construct the retarded and advanced Green's functions $\Delta_0^{R/A}$ using the fact that $DD^*=D^*D=\Box+m^2$. Let $G^{R/A}$ be retarded/advanced Green's function for $(\Box+m^2)$. It can be shown that:
\begin{equation}
\Delta_0^{R/A}(x,y)=\left(\begin{array}{cc}
0&-D^{*}(x)G^{R/A}(x,y)\\
D^T(x)G^{R/A}(x,y)&0
\end{array}\right)\doteq\left(\begin{array}{cc}
0&{K_0}_*^{R/A}\\
-K_0^{R/A}&0
\end{array}
\right)\,,\label{free}
\end{equation}
The Peierls bracket for the free theory $\{.,.\}_0$ is given by equation (\ref{peierls2}) with the causal propagator $\Delta_0=\Delta_0^R-\Delta_0^A$ determined by (\ref{free}).
The interacting action functional for the Gross-Neveu model in D dimensions with $N$ Dirac spinors (colours) takes the form (with the spinor indices suppressed):
\begin{multline}
 S(f)(u)=\\= \int dx\, f(x)\Big(\sum\limits_{a=1}^N\overline{\Psi}^a_x\wedge(i\partial\!\!\!/-m)\Psi^a_{x}+\frac{\lambda g(x)}{2N}\sum\limits_{a,b=1}^N(\overline{\Psi}^a_{x}\wedge\Psi^a_x)\wedge(\overline{\Psi}^b_{x}\wedge\Psi^b_x)\Big)(u)\,,\label{gn}
\end{multline}
where $\lambda$ is a coupling constant and $g\in\calD(\M)$ is the spacetime cutoff for the interaction. Let $\calJ$ be the ideal generated by elements of the form:
\begin{equation}
\Big(D\Psi^a_x+\frac{\lambda g(x)}{N}\sum\limits_{b=1}^N(\overline{\Psi}^b_{x}\wedge\Psi^b_x)\wedge\Psi^a_x\Big)\oplus -\Big(D^*\overline\Psi^b_{x}+\frac{\lambda g(x)}{N}\sum\limits_{b=1}^N(\overline\Psi^b_{x}\wedge\Psi^b_x)\wedge\overline\Psi^a_{x}\Big)\,.
\end{equation}
Then equations of motion are realized on the algebra $\calF/\calJ$. The second derivative reads:
\begin{multline}
\left<S(1)^{(2)}(1\oplus\tilde{u}_1\wedge \tilde{u}_2);h_1,h_2\right>=\\
=\int d^Dx\Big(\sum\limits_{a=1}^N\overline{h_1}^a(x)Dh_2^a(x)+\frac{\lambda g(x)}{N}\sum\limits_{a,b=1}^N\overline{h_1}^a(\overline{u_1}^bu_2^b-\overline{u_2}^bu_1^b)h_2^a-c.c.\Big)\,.\label{gn2}
\end{multline}
One can read off from the above equation the distribution kernel of $S(1)^{(2)}$ to be (written formally): 
\begin{equation}
S(1)^{(2)}_{ab}(x,y)=\delta(x-y)\delta_{ab}\left(\begin{array}{cc}
0&D^{*^T}(x)+F^T(x)\\
-D(x)-F(x)&0
\end{array}\right)\,,
\end{equation}
where $F(x)=\frac{\lambda g(x)}{2N}\sum\limits_{b=1}^N\overline{\Psi}^b_x\wedge\Psi^b_x$. It turns out that $S(1)^{(2)}$ can be inverted and one can also impose support conditions on the inverse. Let $\Delta_I^{R/A}$ denote the inverse whose support satisfies condition (\ref{retarded}) or (\ref{advanced}) respectively. For every $u=\oplus_{p=0}^nu^{(p)}\in\euC(M,V)$ we have:
\begin{equation}
({\Delta_I}^{R/A})_{ab}(x,y)(u)=\delta_{ab}\left(\begin{array}{cc}
0&{K_I}_*^{R/A}(x,y)(u)\\
-K_I^{R/A}(x,y)(u)&0
\end{array}
\right)\,,\label{inter}
\end{equation}
where the corresponding distribution kernels are defined as:
\begin{gather}
K_I^{R/A}(x,y)(u)\doteq \Big(K_0^{R/A}(x,y)+\\
+\sum_{k=1}^{n/2}(-1)^k\int dz_1...dz_kK_0^{R/A}(x,z_1)K_0^{R/A}(z_1,z_2)...\,K_0^{R/A}(z_{k},y)F(z_1)\wedge...\wedge F(z_k)\Big)(u)
\end{gather}
Analogously for ${K_I}_*^{R/A}$. 
We can conclude that the inverse of $S(1)^{(2)}$ exists as a distribution with values in $\calF^+$. The Peierls bracket for the interacting theory $\{.,.\}_I$ is given by equation (\ref{peierls2}) with the causal propagator $\Delta_I=\Delta_I^R-\Delta_I^A$.
\section{Deformation quantization}\label{quant}
Finally we come to the quantization. In this section we want to show how the formalism we introduced for the classical theory fits into the framework of deformation quantization introduced in \cite{BF0,Duetsch:2000nh,DF}. We start with the free theory and introduce the interaction in the perturbative way. Let $S$ be the free action, i.e. $S\in\calF^2$. We assume that $S(1)^{(2)}$ is a strictly hyperbolic operator on $E$. 
Let $\Delta$ be the corresponding causal propagator, i.e. $\Delta=\Delta^R-\Delta^A$. In the first step we consider only very regular elements of $\calF$. Let $\calF_\reg=\prod\limits_{n=0}^\infty\calD(\M^n,V^{\otimes n})^a$.
The quantum algebra is defined by deforming the $\wedge$-product on $\calF_\reg$. We define the star product on  $\calF_\reg$ analogously to \cite{BDF,DF,Kai} but instead of a symmetric tensor product we use $\wedge$. First we introduce 
a graded functional differential operator $\Gamma_{\Delta}:\bigwedge\limits^2\calF_\reg\rightarrow \bigwedge\limits^2\calF_\reg$:
\begin{equation}\label{Gamn}
\Gamma _{\Delta}(F,G)\doteq(-1)^{(|F|+1)}\frac{1}{2}\int\! dxdy\ \Delta(x,y)\cdot F^{(1)}(x)\wedge G^{(1)}(y)\,,
\end{equation}
where  $F$, $G$ are homogenous. Clearly $\Gamma_{\Delta}$ can be extended also to non-homogenous elements of $\calF_\reg$ by linearity. 
Let $\calF_\reg[[\hbar]]$ denote the space of formal power series with coefficients in $\calF_\reg$. It is equipped with the direct product topology induced by the topology of $\calF_\reg$. By a slight abuse of notation we denote these topologies by the same symbol. The $\star$-product is defined as:
\begin{eqnarray}
\star:\bigwedge^2\calF[[\hbar]]&\rightarrow& \calF[[\hbar]]\nonumber\\
F\star G&\doteq&\exp(i\hbar \Gamma_{\Delta})(F,G),\label{star}
\end{eqnarray}
where $\exp(i\hbar \Gamma_{\Delta})$ is a short-hand notation for a formal power series: $\sum\limits_{n=0}^\infty\frac{1}{n!}(\Gamma_{\Delta})^n$ and $\Gamma_{\Delta}^0(F,G)=F\wedge G$. With the star-product (\ref{star}) we can define the commutator as:
\begin{equation}
[F,G]_{\star}\doteq F\star G-(-1)^{|F||G|}G\star F\,.
\end{equation}
Particularly $\Phi(f)=\int dx f_a(x)\Phi^a_x$, a field smeared with a test function is an element of $\calF_\reg^1$. The $\star$-product of two such elements reads:
\begin{equation}
\Phi(f)\star\Phi(g)=\Phi(f)\wedge \Phi(g)+\frac{i\hbar}{2}\left<f,\Delta g\right>\,,
\end{equation}
where $\left<f,\Delta g\right>\doteq\sum\limits_{a,b}\!\int\!\! dxdy f_a(x)\Delta^{ab}(x,y)g_b(y)$. The corresponding (anti-) commutator takes the form:
\begin{equation}
[\Phi(f),\Phi(g)]_\star=i\hbar\left<f,\Delta g\right>\,.
\end{equation}
As an example we can take the free Dirac field. The causal propagator written in the matrix form reads:
\begin{equation}
\Delta=\left(\begin{array}{cc}
0&K\\
-K_*&0
\end{array}
\right)\,.
\end{equation}
From this, it follows that:
\begin{equation}
[\Psi(g),\overline{\Psi}(f)]_{\star}=i\hbar\left<g,Kf\right>_E=-i\hbar\left<K_*g,f\right>_E=[\overline{\Psi}(f),\Psi(g)]_{\star}\,,
\end{equation}
where $f\in \calD(DM),g\in \calD(D^*M)$.
This is the quantized algebra of free fields. Now we want to introduce the interaction. Following \cite{BDF} we use to this end the relative $S$-matrix. 

Firstly we introduce the time ordered product. Let $\Delta^D=\frac{1}{2}(\Delta^R+\Delta^A)$ be the Dirac propagator. One can define a map $\Gamma_{\Delta^D}:\calF_\reg[[\hbar]]\rightarrow\calF_\reg[[\hbar]]$ on homogenous elements by:
\begin{equation}
\Gamma_{\Delta^D}\doteq(-1)^{|F|}\frac{1}{2}\int\! dxdy\,\Delta^D(x,y)\,F^{(2)}(x,y)\,.
\end{equation}
The time-ordering operator is a map $\mathcal{T}:\calF_\reg[[\hbar]]\rightarrow\calF_\reg[[\hbar]]$ defined as a formal power series:
\begin{equation}
\calT(F)\doteq\exp(i\hbar\Gamma_{\Delta^D})F\,.
\end{equation}
The anti-time-ordering operator $\calT^{-1}$ is a formal inverse of $\calT$ defined by: $\calT^{-1}(F)\doteq\exp(-i\hbar\Gamma_{\Delta^D})F$. The time ordered product $\tio$ on $\calT(\calF_\reg)$ is given by:
\begin{equation}
F\tio G\doteq\calT(\calT^{-1}F\wedge\calT^{-1}G)\,.
\end{equation}
Following \cite{BF0,BDF} we introduce now the interaction by the formula of Bogoliubov.
Let $F\in\calF_\reg^+$ be an interaction term, then we define the formal S-matrix as:
\begin{equation}
\mathcal{S}(F)\doteq\sum\limits_{n=0}^\infty \frac{1}{n!}\calT^{\,n}(F,..,F)\doteq\sum\limits_{n=0}^\infty \frac{1}{n!}F\tio ...\tio F\,.
\end{equation}
With the use of $\mathcal{S}(F)$ one can define the interacting algebra along the lines of \cite{BF0,Duetsch:2000nh,BDF}.

Now we want to extend our discussion to more singular elements of $\calF$. In order to do this we need to replace the product $\star$ with an equivalent one, defined by means of a Hadamard solution satisfying the microlocal spectrum condition \cite{Duetsch:2000nh,BF0,BDF}. Since we are considering here only the Minkowski spacetime, we can choose for concreteness the Wightmann 2-point function $\Delta_+=\frac{i}{2}\Delta+\Delta_1$. We replace now $\frac{i}{2}\Delta$ in the definition of the star product with $\frac{i}{2}\Delta+\Delta_1$. This way we obtain an equivalent product $\starH$, which is related to the old one by the transformation $\alpha_{\Delta_1}$, defined as 
\begin{equation}
\alpha_{\Delta_1}(F)\doteq\exp(\hbar\Gamma_{\Delta_1})F\,,
\end{equation}
with
\begin{equation}\label{alphaH}
\Gamma _{\Delta_1}(F)\doteq(-1)^{|F|}\frac{1}{2}\int\! dxdy\ \Delta_1(x,y)F^{(2)}(x,y)\,.
\end{equation}
The relation between star products $\star$ and $\starH$ reads:
\begin{equation}
F\starH G=\alpha_{\Delta_1}(\alpha_{\Delta_1}^{\minus}F\star \alpha_{\Delta_1}^{\minus}G)\,.
\end{equation}
The new star product can now be extended to the elements of $\calF$. Consider now the map $\alpha_{\Delta_1}^{\minus}:\calF_\reg[[\hbar]]\rightarrow\calF_\reg[[\hbar]]$. We equip the domain $\calF_\reg[[\hbar]]$ with topology $\tau_\Xi$ and define a topology $\tau_{\Delta_1}$ on the target space as the finest one that makes $\al_{\Delta_1}^{\minus}$ continuous. Next we embed  $(\calF_\reg[[\hbar]],\tau_\Xi)$ as a dense topological vector space in $(\calF[[\hbar]],\tau_\Xi)$ and take a sequential closure of the target space $(\calF_\reg[[\hbar]],\tau_{\Delta_1})$ with respect to all the sequences $\alpha_{\Delta_1}^{\minus}(F_n)$, where $F_n$ converges to an element of  $\calF[[\hbar]]$ with respect to $\tau_\Xi$. We denote this closure by $\frakA[[\hbar]]$ and from the construction follows that $\alpha_{\Delta_1}^{\minus}:\calF[[\hbar]]\rightarrow\frakA[[\hbar]]$ is continuous. From the microlocal spectrum condition and the condition (\ref{WFcond}) on wavefront sets of functionals in $\calF$ it follows that the star product $\star$ is also continuous with respect to $\tau_{\Delta_1}$, so it can be extended to a product on $\frakA[[\hbar]]$. This way we obtain an involutive algebra $(\frakA[[\hbar]],\star)$.

The situation with the time-ordered product is more complicated. Like in the bosonic case this product is not continuous with respect to the topology $\tau_\Xi$. Nevertheless it is well defined for $F,G\in \frakA$ with disjoint supports. More generally we can define graded symmetric maps $\calT^{\,n}:\calF[[\hbar]]^{\otimes n}\rightarrow \frakA[[\hbar]]$ by means of:
\begin{equation}
\calT^n(F_1,...,F_n):=\alpha_{\Delta_1}^{\minus}(F_1)\tio...\tio \alpha_{\Delta_1}^{\minus}(F_n)\,,
\end{equation}
for $F_1,...,F_n$ with pairwise disjoint supports. The problem of renormalization is now formulated as the problem of extending maps $\calT^n$ to functionals with coinciding supports. From the mathematical point of view it doesn't differ much from the bosonic case, except of the fact that now we deal with graded symmetric in place of symmetric distributions. As shown in \cite{EG} the extension of $\calT^n$  can be defined if the arguments are local functionals.
\section*{Conclusions}
The proper understanding of classical theory of fermions is essential for applying the functional approach in quantum field theory to anticommutiong fields. In this paper we proposed a formalism that generalizes that of \cite{BDF,DF,DF02,DF04,Pedro}  and stays consistent with \cite{TCN,Thomas}, where the functional approach was applied to free Dirac fields. We also showed that our setting can be used to treat fermion-fermion interactions, in the example of the Gross-Neveu model. The algebraic structure we use can also describe more general nonlinear models, since we admit interaction terms that are defined by arbitrary (possibly infinite) power series. This generalizes the approach of \cite{TCN,Thomas}, where only finite sums are admitted.

Our results can be also applied in quantization of gauge theories with the use of BRST method, since ghost fields have to be of fermionic type. Furthermore our approach allows to treat odd and even variables on the equal footing and is therefore very natural to apply to BV formalism. We also provide a notion of an ``odd'' derivative which can be used to make the formal calculations of BRST and BV quantization more precise. The next step is to apply the results concerning odd fields to complete treatment of gauge field theories in the functional approach. This is done in \cite{FR}. 
\section*{Acknowledgements}
I would like to thank K. Fredenhagen for enlightening discussions and remarks. Furthermore I want to thank R. Brunetti, C. Dappiaggi, K. Keller,  P. Lauridsen Ribeiro and J. Zahn for valuable comments. I am also grateful to the Villigst Stiftung for the financial support of my Ph.D. 
\appendix
%
\section{Topologies}\label{tops}
We recall from \cite{Jar} some of the important definitions from topology and distribution theory. Let $E,F$ be locally convex topological vector spaces (lcvs) such that there exists a dual paring $<.,.>:E\times F\rightarrow\R$. $E$ can be regarded as a linear subspace of $\R^F$ and the topology it inherits from the product topology of $\R^F$, is called \textit{the weak topology}, denoted by $\sigma(E,F)$.
\begin{definition}[\cite{Jar}, 3.2]
A subset $U$ of a topological vector space $X$ is called \textit{complete} (\textit{sequentially complete}) if every Cauchy net (sequence) converges in $U$. We say that $X$ is quasi-complete if every closed bounded subset of $X$ is complete.
\end{definition} 
The property of (sequential) completeness and quasi-completeness is inherited by the infinite direct products and infinite direct sums \cite[3.3.5]{Jar}.
\begin{definition}
Let $E,F$ be Hausdorff lcvs and $\euB$ be the family of bounded sets of the completion of $E$ (bornology). Let $\tau_\euB$ be the topology of uniform convergence on bounded sets it induces on $L(E,F)$. We say that $E$ has the (sequential) approximation property if one of the following equivalent conditions holds:
\begin{enumerate}
\item $E'\otimes F$ is (sequentially) dense in $(L(E,F),\tau_\euB)$ for every $F$\,,
\item $E'\otimes E$ is (sequentially) dense in $(L(E,E),\tau_\euB)$\,,
\item $\1_E$ is the $\tau_\euB$-limit of some (sequence) net in $E'\otimes E$.
\end{enumerate}
\end{definition}
This property is also inherited by a direct product of a family of lcvs that are Hausdorff \cite[18.2.4]{Jar}.
Now we shall describe the topology of $\calA$ in more detail. The family of seminorms  for the space $\calE(\M)$ is given by:
\begin{equation}
p_{K,m}(\phi)=\sup_{x\in K\atop |\alpha|\leq m}|\partial^\alpha\phi(x)|\,,
\end{equation}
where $\alpha\in\Na^N$ is a multiindex. A set $B\subset\calE(\M)$ is bounded if $\sup\limits_{\phi\in B}\{p_{K,m}(\phi)\}<\infty$ for all seminorms $p_{K,m}$.
Let $\euB$ be the family of bounded sets in $\calE(\M)$. The \textit{strong topology} on the dual space $\calE'(\M)$ is defined by a family of seminorms: $p_B(T)\doteq\sup_{\phi\in B}\left<T,\phi\right>$, where $B\in \euB$, $T\in \calE'(\M)$, $\phi\in \calE(\M)$. The bounded sets in $\euC(\M,V)$ are finite products of bounded sets in the constituent spaces $\calE(\M^p,V^{\otimes p})$. It follows that the strong topology on $\calA$ is generated by the family of seminorms: $p_{B_{i_1},\ldots,B_{i_m}}(T)\doteq\sum\limits_{k=1}^m\sup\limits_{\phi\in B_{i_k}}\left<T,\phi\right>$, where $k\in\Na$ and $B_{i_k}\subset \calE(\M^{i_k},V^{\otimes i_k})^a$ is bounded.

The function spaces $\mathscr{S}(\M)$, $\calD(\M)$, $\calE(\M)$, as well as their strong duals $\mathscr{S}'(\M)$, $\calD'(\M)$, $\calE'(\M)$, are reflexive complete (therefore quasi-complete) nuclear spaces (\cite{Pie}) and they have the (sequential) approximation property  (\cite{Jar,Sch1}). It is shown in \cite{Nle} that every locally convex vector space is nuclear with its weak topology, so the weak duals $\mathscr{S}'(\M)$, $\calD'(\M)$, $\calE'(\M)$ are trivially nuclear.
Note that spaces $\calE(\M^n,V^{\otimes n})^a$ defined in section \ref{asf} are Frech\'et nuclear spaces and this holds also for their countable direct sum $\euC(\M,V)$. The corresponding dual spaces $\calE'(\M^n,V^{\otimes n})^a$ are nuclear (with strong and weak topologies) and we can equip $\calA=\euC(\R,V)'=\prod\limits_{n=0}^\infty\calE'(\M^n,V^{\otimes n})^a$ with the direct product topology $\tau_b$ (or $\tau_\sigma$) induced by the strong (weak) topologies on each of the factors. Locally convex topological vector space $(\calA,\tau_b)$ (or $(\calA,\tau_\sigma)$) 
 is also nuclear because the nuclearity is preserved under the countable direct product. For our purposes it is sufficient to use the weak topology $(\calA,\tau_\sigma)$. The quasi-completeness is also preserved under taking the direct products, so we conclude that $(\calA,\tau_\sigma)$ is quasi-complete.

\end{document}